\documentclass[10pt]{llncs}

\usepackage{ltl}
\usepackage{wrapfig}
\usepackage{booktabs}
\usepackage{tikz}
\usepackage{color}
\usepackage{scalefnt}
\usepackage{subfigure}

\newcommand{\AP}{\textit{AP} }
\renewcommand{\next}{\LTLcircle}

\newcommand{\nats}{\mathbb{N}}

\newcommand{\pspace}{\textsc{Pspace}}

\newcommand{\expspace}{\textsc{Expspace}}

\newcommand{\bp}[1]{\emph{BP}(#1)}

\begin{document}

\title{Synthesizing Skeletons for Reactive Systems\thanks{This work was partially funded by the European Research Council (ERC) Grant OSARES (No. 683300) and by the Deutsche Telekom Foundation.}}
\author{Bernd Finkbeiner \and Hazem Torfah}
\institute{Saarland University}
\date{}
\maketitle

\begin{abstract}
  We present an analysis technique for temporal specifications of
  reactive systems that identifies, on the level of individual system
  outputs over time, which parts of the implementation are determined
  by the specification, and which parts are still open.  This
  information is represented in the form of a labeled transition
  system, which we call skeleton. Each state of the skeleton is
  labeled with a three-valued assignment to the output variables: each
  output can be true, false, or open, where true or false means that
  the value must be true or false, respectively, and open means that
  either value is still possible.  We present algorithms for the
  verification of skeletons and for the learning-based synthesis of
  skeletons from specifications in linear-time temporal logic
  (LTL). The algorithm returns a skeleton that satisfies the given LTL
  specification in time polynomial in the size of the minimal skeleton. Our new
  analysis technique can be used to recognize and repair
  specifications that underspecify critical situations. The technique
  thus complements existing methods for the recognition and repair of
  overspecifications via the identification of unrealizable cores.
\end{abstract}

\section{Introduction}
The great advantage of synthesis is that it constructs an implementation automatically from a specification -- no programming required. The great disadvantage of synthesis is that the synthesized implementation is only as good as its specification, and writing good specifications is extremely difficult.

Roughly speaking, there are two fundamental errors that can happen when writing a specification. The first type of error is to \emph{overspecify} the system such that actually no implementation exists anymore. This type of error can be found by a synthesis algorithm (it fails!), and synthesis tools commonly assist in the repair of such errors by identifying an unrealizable core of the specification (cf. \cite{DBLP:conf/fmcad/AlurMT13,Koenighofer2011,DBLP:conf/memocode/LiDS11}).
The second type of error is to \emph{underspecify} the system such that not all implementations that satisfy the specification actually perform as intended. This type of error is much harder to detect. The synthesis succeeds, and even if we convince ourselves that the synthesis tool has actually chosen an implementation that performs as intended, there is no guarantee that this will again be the case when a new implementation is synthesized from the same or an extended specification.

The underlying problem is that synthesis algorithms have the freedom to resolve any underspecified behavior in the specification, and we have no way of knowing which parts of the behavior were fixed by the specification, and which parts were chosen by the synthesis algorithm.

In this paper, we introduce a new artifact that can be produced by synthesis algorithms and which provides exactly this information. We call this artifact the \emph{skeleton} of the specification. We envision that synthesis algorithms would produce the skeleton along with the actual implementation, so that the user of the algorithm understands where the implementation is underspecified, and can, if so desired, strengthen the specification in critical areas.

A skeleton is a labeled transition system defined over three-valued sets of atomic propositions, where in each state of the skeleton an atomic proposition is either \emph{true}, \emph{false}, or \emph{open}. For a given specification, the truth value of a proposition in some state of the skeleton is \emph{open} if it can be replaced by $\mathit{true}$ as well as by $\mathit{false}$ without violating the specification. Consider for example the LTL formula $\next p$ for some atomic proposition $p$. Any transition system that satisfies the formula has truth value $\mathit{true}$ for  $p$ in the second position of every path of the transition system. On the other hand, whether $p$ is $\mathit{true}$ or $\mathit{false}$ in the initial state is not determined, either truth value would work. In this case, the skeleton would not fix a particular truth value, but rather leave the value of $p$ in the initial state open. In a sense, the skeleton implements only those parts of the transition system that are determined by the specification.

Skeletons are useful to understand the meaning of partially written specifications. Consider, for example, an arbiter over two clients that share some resource. Each client can make a request to the source (via the inputs $r_1$ and $r_2$) and the arbiter can, accordingly, decide to give out grants via the outputs $g_1$ and $g_2$. A specification for the arbiter might begin with the property of mutual exclusion, i.e., the LTL formula $\LTLsquare(\bar g_1 \vee \bar g_2)$ stating that  only one of the clients should have access to the resource at a time.   Figure \ref{fig:arbiter} shows an implementation of this specification as a transition system and  a skeleton.  The transition system has a single state, and no grants are given at any time (see Figure \ref{subfig:ts1}). The skeleton shown in Figure~\ref{subfig:ske1}  reveals that all outputs are open, as indicated by the question mark.
If we extend the specification with the property $\bar g_1 \wedge \bar g_2$, then the previous transition system does not need to change, because it already satsifies the extended specification. The skeleton, on the other hand, now indicates that the output in the initial state is determined. The output in subsequent states is still open (see Figure \ref{subfig:ske2}). Extending the specification further with the property $\LTLsquare(r_1 \rightarrow \next g_1)$ results in a skeleton where the responses to requests from the first client are determined, and outputs in situations where there is no request from the first client are still open (see Figure \ref{subfig:ske3}). An implementation for this specification could be the transition system that never gives a grant to the second client (see Figure \ref{subfig:ts2}).

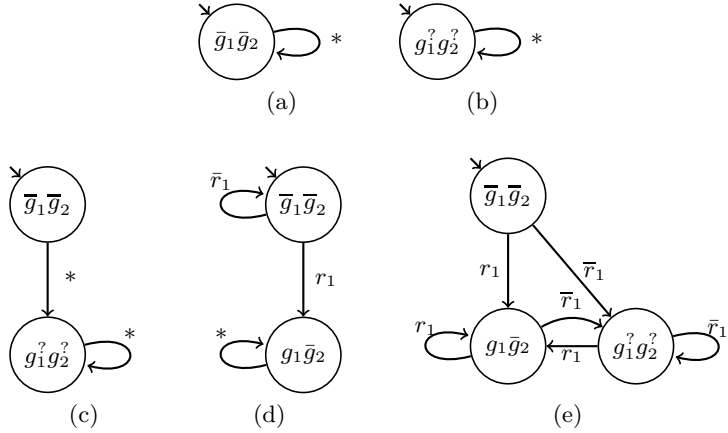
\begin{figure}[t]
\centering
	\subfigure[][]{\label{subfig:ts1}
        \begin{tikzpicture}[inner sep=1pt, minimum size =1cm,node distance=2cm,semithick]
        	\node(s0)[draw,circle] at(0,0){$\bar g_{1} \bar g_2$};
        	\draw[->,thick](-.5,.5) --(s0);
        	\path[->,thick](s0) edge [loop right] node[ right=-8] {*}(s0);
        \end{tikzpicture}   
        } 
      \subfigure[][]{\label{subfig:ske1}
        \begin{tikzpicture}[inner sep=1pt, minimum size =1cm,node distance=2cm,semithick]
        	 \node(s1)[draw,circle] at(0,-1.5){$g_{1}^? g_2^?$};
        	\draw[->,thick](-.5,-1) --(s1);
        	\path[->,thick](s1) edge [loop right] node[ right=-8] {*}(s1);
        \end{tikzpicture}   
        }\\ 
      \subfigure[c][]{\label{subfig:ske2}
        \begin{tikzpicture}[inner sep=1pt, minimum size=1cm,node distance=2cm,semithick]
        	\node(s0)[draw, circle] at(0,0){$\overline g_{1}\overline g_2$};
        	\node(s1)[draw, circle] at(0,-2){$g_{1}^? g_2^?$};
        	\draw[->,thick](-.5,.5) --(s0);
           	\path[->,thick](s0) edge node [right=-6]{*}(s1);
        	\path[->,thick](s1) edge [loop right] node [above=-8]{*}(s1);
        \end{tikzpicture}
        }
     \subfigure[]{\label{subfig:ts2}
       \begin{tikzpicture}[draw,inner sep=1pt, minimum size=1cm,node distance=2cm,semithick]
         \node(s0')[draw, circle] at(0,2){$\overline g_{1}\overline g_2$};
        	\node(s1')[draw, circle] at(0,0){$g_{1} \bar g_2$};
        	\draw[->,thick](-.5,2.5) --(s0');
           	\path[->,thick](s0') edge node [right=-6]{$r_1$}(s1');
        	\path[->,thick](s1') edge [loop left] node [above=-8]{*}(s1');
        	\path[->,thick](s0') edge [loop left] node [above=-6]{$\bar r_1$}(s0');
		\end{tikzpicture}
		}
         \subfigure[]{\label{subfig:ske3}
        \begin{tikzpicture}[draw,inner sep=1pt, minimum size=1cm,node distance=2cm,semithick]
        	\node(s0)[draw, circle] at(0,0){$\overline g_{1}\overline g_2$};
        	\node(s1)[draw, circle] at(0,-2){$g_{1}\bar g_2$};
        	\node(s2)[draw, circle] at(1.7,-2){$g_{1}^?g_2^?$};
        	\path[->,thick](s0) edge node [right=-6]{$\overline  r_1$}(s2);
        	\path[->,thick](s2) edge [loop right] node [above=-8]{$\bar r_1$}(s2);
        	\path[->,thick](s1) edge [loop left] node [above=-8]{$r_1$}(s1);
        	\path[->,thick](s0) edge node [left=-8]{$r_1$}(s1);
        	\path[->,thick](s2) edge node [below= -11]{$r_1$}(s1);
        	\path[->,thick](s1) edge [bend left =30] node [above=-8] {$\overline r_1$}(s2);
        	\draw[->,thick](-.5,.5) --(s0);
         \end{tikzpicture}
         }
\label{fig:arbiter}
\caption{Transition systems and skeletons for an arbiter specification. The symbol * denotes all possible input labels. }	
\end{figure}

We study the \emph{model checking} and \emph{synthesis} problems for skeletons. For a given LTL formula $\varphi$ and a skeleton $\mathcal S$ we say that $\mathcal S $ is a model of the LTL formula $\varphi$, if each trace in $\mathcal S$ satisfies following condition: If the truth value for some proposition $p$ in some position of the trace is open, then $\varphi$ must both have a model where $p$ is $\mathit{true}$ at this position, and a model where $p$ is $\mathit{false}$ at this position. Furthermore, if the trace has truth value \emph{true} or \emph{false} for $p$ at some position, then \emph{all} models of $\varphi$ map $p$ to the truth value \emph{true} or \emph{false}, respectively, at this position.

 We show that given an LTL formula $\varphi$ we can build a nondeterministic automaton that accepts a sequence over the three-valued semantics if it satisfies the satisfaction relation described above. The automaton is of doubly-exponential size in the length of the formula $\varphi$. With this automaton, the model checking problem can be solved  in $\expspace$.  

To solve the synthesis problem, we could determinize the automaton and check whether there is a skeleton for the formula, along the lines of standard synthesis~\cite{rosner1991modular}, but this construction would be very expensive. Instead, 
 we introduce a synthesis algorithm for skeletons based on \emph{learning}. We show that for each LTL formula, a skeleton that models the formula defines a safety language that can be learned using the learning algorithm L$^*$. The algorithm can learn a skeleton for an LTL formula in time polynomial in the size of the minimal skeleton for the specification. The membership and equivalence queries of the L$^*$ algorithm are answered by the model checking algorithm introduced in this paper.
 
 \paragraph{Related Work.}
 There is a rich body of work on the synthesis of reactive systems from logical specifications~\cite{JSL:9119904,Bloem:2012:SRD:2160191.2160503,ScheweFinkbeiner2013,Manna:1984:SCP:357233.357237,Pnueli:1989:SRM:75277.75293}. Supplemented by many works that investigated the  optimization of specification for synthesis and the identification of unrealizable specification \cite{Koenighofer2011,DBLP:conf/memocode/LiDS11,DBLP:conf/fmcad/AlurMT13}. 
Multi-valued extensions of logics have been rather popular in the verification of systems, where a simple truth value is not enough to determine the quality of implementations. Chechik et. al. provide a theoretical basis for multi-valued model checking \cite{Chechik:2003:MSM:990010.990011}, where the satisfaction relation $\mathcal M \models \varphi$ for a model $\mathcal M$ and a specification $\varphi$ can be multi-valued. Bruns and Godefroid experiment on multi-valued logics and show that many algorithms for multi-valued logics can be reduced to ones for two-valued logics \cite{Bruns2004}. Easterbrook and Chechik introduce a framework where multiple inconsistent models are merged according to an underlying specification given in a multi-valued logic, where the different values in the specification represent the different levels of uncertainty, priority and agreement between the merged models \cite{Easterbrook:2001:FMR:381473.381516}. In comparison to all these works, we are interested in multi-valued extensions of the models themselves and in the synthesis of such models, in order to determine the amount of information that resides in a specification.

The term skeleton has been also used by Emerson and Clarke which shall not be confused with the skeletons presented here. They presented a method for the synthesis of synchronization skeletons that abstract from details irrelevant to synchronization of concurrent systems \cite{Clarke:1981:DSS:648063.747438}. In our skeletons, we stick to the structure of transition systems and leave place holders for the underspecified details, which may then be supplemented with further steps to a complete transition system.

\section{Preliminaries}

\paragraph{Alternating Automata. }
We define an \emph{alternating B\"uchi automaton} as a tuple $\mathcal{A} =
(\Sigma,Q,q_0,\delta,F)$, where $\Sigma$ denotes a
finite alphabet,
$Q$ denotes a finite set of states, $q_0 \in Q$ denotes a designated
initial state, $\delta: Q \times \Sigma \rightarrow
\mathbb{B}^+(Q)$ denotes a transition function, that maps a state and an input letter to a
positive boolean combination of states, and finally the set $F \subseteq Q$ of accepting states.

We define infinite words over $\Sigma$ as sequence $\sigma: \nats \rightarrow \Sigma$. A $\Sigma$-tree is a pair $(\mathcal T, r)$ over a set of directions $D$, where $\mathcal T$ is a prefix-closed subset of $D^*$ and $r: \mathcal T \rightarrow \Sigma$ is a labeling function. The empty sequence $\epsilon$ is called the \emph{root}. The children of a node $n \in \mathcal T$ are nodes $C(n)=\{n \cdot d \in \mathcal T \mid d \in D\}$. 

 A \emph{run} of an automaton $\mathcal{A} = (\Sigma,Q,q_0,\delta,F)$ on a sequence $\sigma: \nats \rightarrow \Sigma$ is a $Q$-tree $(\mathcal T,r)$ with $r(\epsilon)=q_0$ and for all nodes $n\in \mathcal T$, if $r(n)=q$ then the set $\{r(n') \mid n' \in C(n)\}$ satisfies $\delta(q,\sigma(|n|))$. 

A run $(\mathcal T,r)$ is \emph{accepting} if for every infinite branch $n_0,n_1, \dots$ the sequence 
$r(n_0)r(n_1)\ldots$ satisfies the
\emph{B\"uchi condition}, which requires that some state from $F$ occures
infinitely often in the sequence $r(n_0)r(n_1)\ldots$. 

The set of accepted words by the automaton $\mathcal A$ is the language of the automaton and is denoted by $L(\mathcal A)$. An automaton is empty iff its language is the empty set.

A \emph{nondeterministic} automaton is a special alternating
automaton, where the image of $\delta$ consists only of such formulas
that, when rewritten in disjunctive normal form, contain exactly one
element of $Q$ in
every disjunct.

An alternating automaton is called \emph{universal} if, for all states $q$ and input letters $\alpha$, $\delta(q,\alpha)$ is a conjunction.
A  universal and nondeterministic automaton is called \emph{deterministic}.

A  \emph{B\"uchi} automaton is called a \emph{safety} automaton if $Q= F$. Safety
automata are denoted by a tuple $(\Sigma,Q,q_0,\delta)$. 
For safety automata, every run graph is accepting. 

The dual of B\"uchi automata are \emph{co-B\"uchi} automata. In a co-B\"uchi automaton the set $F$ is a set of rejecting states and a run is accepting if it has only finitely many appearances of states in $F$. 

\paragraph*{Safety Languages:} A finite word $w= \{1, \dots , i\} \rightarrow \Sigma$ over some finite alphabet $\Sigma$ is called a bad-prefix for a language $L \subseteq \Sigma^\omega$, if every infinite word $\sigma \in (\nats \rightarrow \Sigma)$ with prefix $ w$ is not in the language $L$. 
A language $L \subseteq (\nats \rightarrow \Sigma)$ is called a safety language, if  every $\sigma \not \in L$ has a bad-prefix.
We denote the set of bad-prefixes for a language $L$ by $\bp{L}$. 
For every safety language $L$ we can define a finite word automaton $\mathcal B = ( Q_{\mathcal B}, Q_{\mathcal B,0}, F_{\mathcal B},  \delta_{\mathcal B} )$ that accepts the language $\bp{L}$. We call $\mathcal B$  the  bad-prefix automaton of $L$.

\paragraph{Linear-time Temporal Logic:}
We use Linear-time Temporal Logic (LTL) \cite{Pnueli:1977:TLP:1382431.1382534}, with the usual temporal operators Next $\LTLcircle$, Until $\cal U $ and the derived operators Eventually $\LTLdiamond$ and Globally~$\LTLsquare$. 
 LTL formulas are defined over a set of atomic propositions $\AP = I \cup O$, which is partitioned into a set $I$ of input propositions and a set $O$ of output propositions. 
We denote the satisfaction of an LTL formula $\varphi$ by an infinite sequence $\sigma\colon \mathbb N \rightarrow 2^{\AP}$ of valuations of the atomic propositions  by $\sigma \models \varphi$. For an LTL formula $\varphi$ we define the language $L(\varphi)$ by the set $\{\sigma \in(\nats \rightarrow 2^\AP ) \mid \sigma \models \varphi \}$. 

\paragraph{Implementations:}We represent implementations as \textit{labeled transition systems}. For a given finite set $\Upsilon$ of directions and a finite set $\Sigma$ of labels, a $\Sigma$-labeled $\Upsilon$-transition system is a tuple $\mathcal T = (T,t_0,\tau,o)$, consisting of a finite set of states~$T$, an initial state $t_0\in T$, a transition function $\tau\colon T\times \Upsilon \rightarrow T$, and a labeling function $o\colon T\rightarrow \Sigma$.
A \textit{path} in $\mathcal T$ is a sequence~$\pi\colon \mathbb N \rightarrow T \times \Upsilon$  of states and directions that follows the transition function, i.e., for all $i \in \mathbb N$  if $ \pi(i) = (t_i,e_i)$ and $\pi(i + 1) = (t_{i+1}, e_{i+1})$, then $t_{i+1} = \tau(t_i, e_i)$. We call a path initial if it starts with the initial state: $\pi(0) = (t_0 , e )$ for some $e \in \Upsilon$. We denote the set of initial paths of $\mathcal T$ by $\emph{Path}(T)$. For a path $\pi \in \emph{Path}(T)$, we  denote the sequence $\sigma_{\pi}\colon i \mapsto o(\pi(i))$, where $o(t,e)= (o(t)\cup e)$ by the \emph{trace} of $\pi$.  We call the set of traces of the paths of a transition system $\mathcal T$ the language of the $\mathcal T$, denoted by $L(\mathcal T)$. 

For a set of atomic propositions $\AP= O \cup I$, we say that a $2^O$-labeled $2^I$-transition system $\mathcal T$ satisfies  an LTL formula $\varphi$, if and only if $L(T) \subseteq L(\varphi)$, i.e., every trace of $\mathcal T$ satisfies $\varphi$. In this case we call $\mathcal T$ a model of $\varphi$. 

\paragraph{Multi-valued Sets:} A multi-valued set over an alphabet $\Sigma $ and set of values $\Gamma$ is a function $v\in (\Sigma \rightarrow \Gamma)$. The simplest type of multi-valued sets is the two-valued set which define the notion of sets as we know, where $\Sigma$ is a set of symbols and $\Gamma=\{\bot, \top\}$, i.e., for a two-valued set $v$ over $\Sigma$ and $\Gamma$, a symbol $a\in \Sigma$ is in $v$ if $v(a)=\top$, and not otherwise. The set of all multi-valued sets over an alphabet $\Sigma$ and a set of values $\Gamma$ is denoted by $\Gamma^\Sigma$, e.g., in the usual set notion this is the set $\{\bot,\top\}^\Sigma$ or as we know it $2^\Sigma$ for an alphabet $\Sigma$.

For a multi-valued set $v \in \Gamma^\Sigma$ and for $p \in \Sigma$ and $h \in \Gamma$ we define the multi-valued set $v' = v[p \mapsto h]$, where $v'(p)=h$ and for all $p' \in \Sigma \setminus \{p\}$, we have $v'(p')=v(p')$. 
For a multi-valued set $v \in \Gamma^\Sigma$ and for a set $\Sigma' \subseteq \Sigma$ the set $v_{\Sigma'}\in \Gamma^{\Sigma'}$ is the multi-valued set obtained by projection from $\Sigma$ to $\Sigma'$. 
\section{Skeletons}

 An \emph{open set} over an alphabet $\Sigma$ is a three-valued set $v:\{\top,\bot,?\}^\Sigma$, where each element $a \in\Sigma$ is either in $v$ denoted by $v(a)= \top$, not in $v$ denoted by $v(a)= \bot$, or it is open whether it is in the set or not, i.e., it could be one of both, denoted by $v(a)=?$. In the remainder of the paper, we denote the set $\{\top,\bot,?\}^\Sigma$ by $3^\Sigma$. For  two open sets $v,v' \in 3^\Sigma$ we define the partial order $\sqsubseteq$ such that $v \sqsubseteq v'$ if and only if for all symbols $a\in \Sigma$, $v(a) \preceq v'(a)$ with respect to the lattice~$\preceq =\{(\bot,\bot),(\top,\top),(\bot,?), (\top,?),(?,?)\}$.  

We call a sequence $\sigma$ an open sequence if it is a sequence over open sets, i.e., $\sigma \in ( \nats \rightarrow 3^\Sigma)$. For two open sequences $\sigma$ and $\sigma'$ we define the partial order $\sqsubseteq$ such that $\sigma \sqsubseteq \sigma'$ if for all $i \in \nats$, $\sigma(i) \sqsubseteq \sigma'(i)$. For a sequence $\sigma \in (\nats \rightarrow 3^\Sigma)$ and $\Sigma' \subseteq \Sigma$ the sequence $\sigma_{\Sigma'} \in (\nats \rightarrow 3^{\Sigma'})$ is the sequence where for all  $i$, $\sigma_{\Sigma'}(i)= \sigma(i)_{\Sigma'}$. 

We define the satisfaction relation of LTL over open sequences as follows. Given an LTL formula $\varphi$ over a set of atomic propositions $\AP=O\cup I$, an open sequence $\sigma$ satisfies $\varphi$, denoted by $\sigma \models \varphi$, if for each sequence $\sigma' \in L(\varphi)$ that is input equivalent to $\sigma$, i.e., $\sigma_I = \sigma'_I$, we have $\sigma' \sqsubseteq \sigma$. For a fixed sequence of inputs $\varsigma \in (\nats \rightarrow 2^I)$, there is a unique open sequence $\sigma$ with $\sigma_I=\varsigma$ that satisfies $\varphi$ and that is minimial with respect to the partial order $\sqsubseteq$, i.e., for all sequences $\sigma'\in (\nats \rightarrow 3^\AP)$ with $\sigma' \models \varphi$ and $\sigma'_I =\varsigma$, we have $\sigma \sqsubseteq \sigma'$. We call such sequence a \emph{minimal satisfying sequence}. For an LTL formula $\varphi$, we denote the set of all minimal satisfying sequences by $\min(\varphi)$. 

Building on the definitions of open sequences and transition systems we introduce the notion of \emph{skeletons} of reactive systems, which are transition systems labeled with open sets from $3^O$.

\begin{definition}[Skeleton]
 	For a set $\AP=O\cup I$ of atomic propositions, a \emph{skeleton} over $\AP$ is a $3^O$-labeled-$2^I$-transition system.
\end{definition}	
 
The language of a skeleton $\mathcal S$ is the set of open sequences given by the set of its traces.
Figure~\ref{fig:skeletons} shows four skeletons defined over the sets $I=\{r_1,r_2\}$ and $O=\{g_1,g_2\}$. Figures \ref{subfig:1} and \ref{subfig:3} both define the language $\{\sigma: \nats \rightarrow 3^\AP \mid \forall i. \sigma(i)(g_1)=\sigma(i)(g_2) =?  \}$, i.e., for all input sequences the values of the output propositions $g_1$ and $g_2$ are open in all positions. The language of the skeleton in Figure \ref{subfig:2} is the set $\{\sigma: \nats \rightarrow 3^\AP \mid \sigma(0)(g_1)=\sigma(0)(g_2)=\bot, \forall i>0.~  \sigma(i)(g_1)=\top \wedge \sigma(i)(g_2)=? \}$ where the values of $g_1$ are fixed in all positions and for $g_2$  only in the first position of the sequence.\footnote{Note that skeletons have no open values for input propositions.} 

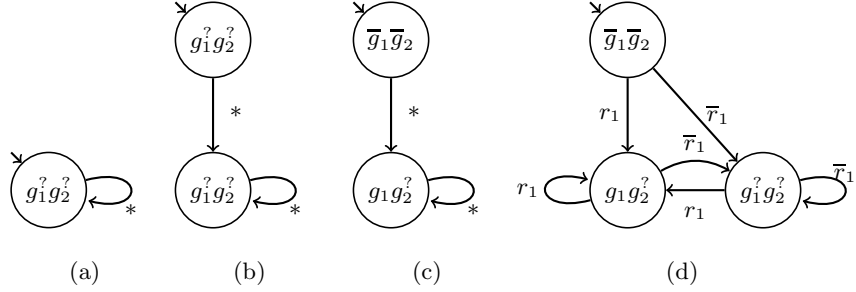
\begin{figure}[t]
\centering
	\subfigure[]{\label{subfig:1}
        \begin{tikzpicture}[inner sep=1pt, minimum size =1cm,node distance=2cm,semithick]
        	\node(s0)[draw,circle] at(0,0){$g_{1}^? g_2^?$};
        	\draw[->,thick](-.5,.5) --(s0);
        	\path[->,thick](s0) edge [loop right] node[ below=-6] {*}(s0);
        \end{tikzpicture}   
        }\hspace{-7pt}
	\subfigure[]{\label{subfig:3}
        \begin{tikzpicture}[inner sep=1pt, minimum size=1cm,node distance=2cm,semithick]
        	\node(s0)[draw, circle] at(0,0){$g_{1}^? g_2^?$};
        	\node(s1)[draw, circle] at(0,-2){$g_{1}^? g_2^?$};
        	\draw[->,thick](-.5,.5) --(s0);
           	\path[->,thick](s0) edge node [right=-6]{*}(s1);
        	\path[->,thick](s1) edge [loop right] node [below=-6]{*}(s1);
        \end{tikzpicture}
        }
\hspace{-5pt}
\subfigure[]{\label{subfig:2}
        \begin{tikzpicture}[inner sep=1pt, minimum size=1cm,node distance=2cm,semithick]
        	\node(s0)[draw, circle] at(0,0){$\overline g_{1}\overline g_2$};
        	\node(s1)[draw, circle] at(0,-2){$g_{1} g_2^?$};
        	\draw[->,thick](-.5,.5) --(s0);
           	\path[->,thick](s0) edge node [right=-6]{*}(s1);
        	\path[->,thick](s1) edge [loop right] node [below=-6]{*}(s1);
        \end{tikzpicture}
        }
	\hspace{-20pt}
	\subfigure[]{\label{subfig:4}
        \begin{tikzpicture}[draw,inner sep=1pt, minimum size=1cm,node distance=2cm,semithick]
        	\node(s0)[draw, circle] at(0,0){$\overline g_{1}\overline g_2$};
        	\node(s1)[draw, circle] at(0,-2){$g_{1}g_2^?$};
        	\node(s2)[draw, circle] at(1.8,-2){$g_{1}^?g_2^?$};
        	\path[->,thick](s0) edge node [right=-6]{$\overline  r_1$}(s2);
        	\path[->,thick](s2) edge [loop right] node [above=-8]{$\overline r_1$}(s2);
        	\path[->,thick](s1) edge [loop left] node [left=-8]{$r_1$}(s1);
        	\path[->,thick](s2) edge node [below=-6]{$r_1$}(s1);
        	\path[->,thick](s0) edge node [left=-8]{$r_1$}(s1);
        	\path[->,thick](s1) edge [bend left] node [above=-8] {$\overline r_1$}(s2);
        	\draw[->,thick](-.5,.5) --(s0);
        \end{tikzpicture}
        }
\caption{Skeletons over the sets $I=\{r_1,r_2\}$ and $O=\{g_1,g_2\}$}
\label{fig:skeletons}
\end{figure}	

We say that a skeleton $\mathcal S$ is a \emph{model} of an LTL formula $\varphi$ denoted by $\mathcal S \models \varphi$, if $L(\mathcal S) = \min(\varphi)$. 
Intuitively, for an LTL formula $\varphi$, a skeleton gives an incomplete transition system where values of atomic propositions that are not deterministically fixed by $\varphi$, are left open, i.e., they are mapped to the value $?$ in the open set of a state.   
Consider the formula $\varphi = \overline g_1 \wedge \overline g_2 \wedge \LTLsquare (r_1 \rightarrow \next g_1 )$. We notice that all transition systems that satisfy $\varphi$ must have the label $\overline g_1 \overline g_2$ in the initial state. For the rest of the transition system, the formula forces only to label a state with $g_1$ in case the direction(input) leading to this state contains the proposition $r_1$, and leaves it open on how to label the states reached by other directions, or whether to label a state with $g_2$ if it is reached by an input where $r_1$ is true (Figure \ref{subfig:4}). 

Building on the satisfaction relation between LTL and skeleton we investigate in the next sections the problems of model checking and synthesis of skeletons. 
\section{Model Checking Skeletons}

We present an automata-based model checking algorithm for skeletons. Given an LTL formula $\varphi$ we show that we can construct a nondeterministic B\"uchi automaton that recognizes the complement language $\overline{\min(\varphi)}$. Using the usual product construction, in this case, the product of the automaton and the skeleton, one can check whether the resulting automaton contains a path that simulates an accepting path in the nondeterministic automaton. If this is the case, then the language of the skeleton contains a sequence in $\overline{\min(\varphi)}$ and, thus, the skeleton is not a model for the formula $\varphi$. Using the construction of the product automaton we also show that checking  whether a skeleton is a model of an LTL formula can be done in space exponential in the length of the formula.

\begin{lemma}
	Given an LTL formula $\varphi$ we can build a nondeterministic B\"uchi automaton $\mathcal N =(3^\AP,Q,q_0,F,\delta)$ such that $L(\mathcal N)=\overline{\min(\varphi)}$. The number of states of $\mathcal N$ is doubly-exponential in the length of $\varphi$. 
\label{lem:NBAforLTL}
\end{lemma}

\paragraph{Construction.} 
The language $\overline{\min(\varphi)}$ contains all sequences $\sigma:\nats \rightarrow 3^\AP$ that are not minimal satisfying open sequences for $\varphi$. These can be distinguished by two types of open sequences. The first type involves sequences $\sigma$ where in some position $i$ the truth value of a proposition $p \in \AP$ is open (mapped to $?$), although, in all   sequences $\sigma' \in L(\varphi)$ with $\sigma_I = \sigma'_I$ the proposition $p$ has the one  same truth value (one of $\top$ or $\bot$ in all sequences) at position $i$. The second type are  sequences $\sigma$, where in some position $i$ a proposition $p$  has truth value $\bot$(resp. $\top$), although, there exists another sequence $\sigma' \in L(\varphi)$ with $\sigma'_I=\sigma_I $ and $\sigma'(i)(p)=\top$(resp. $\bot$).
The latter case also subsumes the case of sequences $\sigma \in (\nats \rightarrow 2^\AP)$ with $\sigma \not \in L(\varphi)$. 

We construct a B\"uchi automaton $\mathcal N =(3^\AP,Q,q_0,F,\delta)$ that accepts an open sequence $\sigma$ if and only if $\sigma \not \in \min(\varphi)$. The automaton is composed of two nondeterministic B\"uchi automata $\mathcal N_1=(3^\AP,Q_1,q_{0,1},F_1,\delta_1)$ and $\mathcal N_2=(3^\AP,Q_2,q_{0,2},F_2,\delta_2)$, one for each of the sequence types mentioned above. We define the automaton as $\mathcal N = \mathcal N_1 \vee \mathcal N_2$, where  $Q=\{q_0\} \cup Q_1 \cup Q_2$, $F= F_1 \cup F_2$ and $\delta= \{(q_0,a, \delta_1(q_{0,1},a) \vee \delta_2(q_{0,2},a))\mid a \in 3^\AP\}\cup \delta_1 \cup \delta_2$

Automaton $\mathcal N_1$ accepts a sequence $\sigma \in (\nats \rightarrow 3^\AP)$ if $\sigma$ has a position $i$ where an atomic proposition $p\in \AP$ is incorrectly marked as open. The automaton $\mathcal N_1$ can be constructed as follows:

Let $\mathcal U_1 = (2^\AP, Q_1^{\mathcal U},q_{0,1}^{\mathcal U}, F_{1}^{\mathcal U}, \delta_1^{\mathcal U})$ be a universal co-B\"uchi automaton for the formula $\neg \varphi$. We extend the automaton $\mathcal U_1$ to another universal co-B\"uchi automaton $\mathcal U_1^*$ over an extended alphabet $\{\top,\bot,?,*_\top, *_\bot\}^\AP$. We make use of the values $*_\top$ and $*_\bot$ to encode in the input sequence whether a mapping to $?$ is wrong, and whether it is wrong when replacing $?$ by $\top$ or by $\bot$. We define $\mathcal U_1^* = (\{\top,\bot,?,*_\top, *_\bot\}^\AP, Q_1^*,q_{0,1}^*,F_1^*, \delta_1^*)$ over two copies of the automaton $\mathcal U_1$(denoted by the numbers 1 and 2) where $Q_1^* = Q_1^{\mathcal U} \times \{1,2\}$,~ $q_{0,1}^* = (q_{0,1}^{\mathcal U},1)$,~ $F_1^* = F_1^{\mathcal U}\times \{1,2\}$. The transition function $\delta_1^*$ is given by the union of the following sets:
	\begin{itemize}
		\item $ \{ ((q,h),v, \delta_1^{\mathcal U}(q,v)_{\{q'\in Q_1^{\mathcal U}/(q',h)\}}) \mid h \in \{1,2\},  \forall p \in O. v(p) \in \{\top,\bot\}  \}$ \\  where in both copies of the automaton $\mathcal U_1$, transitions over symbols $v$ with no open values remain in the same copy and follow the structure of the transition relation $\delta_1^{\mathcal U}$ of $\mathcal U_1$. The operation $\{q'\in Q_1^{\mathcal U}/(q',h)\}$ substitutes every appearance of a state $q'$ in $\delta_1^{\mathcal U}(q,v)$ by a state $(q',h)$ from $Q_1^*$. 
		\item 
		   $\{((q,h),v, (\delta_1^{\mathcal U}(q,v[p \mapsto \top])\wedge \delta_1^{\mathcal U}(q,v[p \mapsto \bot]))_{\{q' \in Q_1^{\mathcal U}/(q',h)\}}) \mid\\
			 h \in \{1,2\},  p \in O, v(p)=?, 
		         \}$\\
		     universal transitions for symbols where a proposition $p$ has an open truth value imitating transitions for both truth values $\top$ and $\bot$ for $p$.
		 \item
		    $\{((q,1),v, \delta_1^{\mathcal U}(q, v[p \mapsto \top])_{\{q'\in Q_1^{\mathcal U}/(q',2)\}}) \mid p \in O, v(p)= *_\top    \}$\\
		    when we guess at some position $i$ that an open truth value for a proposition $p$ is wrong, and it is wrong when replacing it by $\top$ we follow the transition $\top$ to the second copy of $\mathcal U_1$ in which $?,*_\bot$ and $*_\top$ are treated equivalently. This helps to check, whether replacing $?$ by $\top$ results in accpeting run in $\mathcal U_1$, which means that at  position $i$ the truth value $\top$ violates the property $\varphi$, and thus it cannot be open at the that point. 
		 \item
		    $\{((q,1),v, \delta_1^{\mathcal U}(q, v[p \mapsto \bot])_{\{q'\in Q_1^{\mathcal U}/(q',2)\}} ) \mid  p \in O, v(p)= *_\bot  \} $\\
		    which introduce transitions that involve the dual case of $*_\top$.
		 \item
		     $\{((q,2),v, (\delta_1^{\mathcal U}(q,v[p \mapsto \top])\wedge \delta_1^{*}(q,v[p \mapsto \bot]))_{\{q'\in Q_1^{\mathcal U}/(q',2)\}}) \mid \\ p \in O, v(p)\in \{*_\bot, *_\top\}   \}$\\
		     these transitions make sure that when moving to copy 2 of $\mathcal U_1$, values $*_\top$ and $*_\bot$ are treated equally to $?$, because after guessing that a $?$ is wrong it must be wrong for all continuations.  
	\end{itemize}    
	
In order to obtain the desired automaton $\mathcal N_1$ over the alphabet $3^\AP$ we first transform the automaton $\mathcal U_1^*$ to a nondeterministic automaton $\mathcal N_1^*$ with $L(\mathcal U_1^*)= L(\mathcal N_1^*)$ using a subset construction. This is necessary in order to merge all transitions $*_\bot$ at one level into one state. The same holds also for transitions $*_\top$. In this way, we can check whether at some position in a sequence a value $?  $ is wrong by checking all possible branches of the automaton $\mathcal U_1^*$ at that level. The automaton $\mathcal N_1^*$ can be transformed now to the desired automaton $\mathcal N_1$ by projecting every transition label with values in $\{*_\top,*_\bot\}$ to a label $v' \in 3^\AP$ such that for every $p \in O$, if $v(p)=*_\top$ or $v(p)=*_\bot$ then $v'(p)=?$.  

The size of the automaton $\mathcal U_1$ is exponential in the length of $\varphi$ using the transformation of LTL formulas into alternating B\"uchi automata \cite{Vardi95alternatingautomata}, and then using a subset construction. The transformation to $\mathcal U_1^*$ from  $\mathcal U_1$, and to $\mathcal N_1$ from $\mathcal N_1^*$ are both polynomial, and exponential from $\mathcal U_1^*$ to $\mathcal N^*_1$. Thus, the size of $\mathcal N_1$ is doubly-exponential in the length of $\varphi$. 

In a similar way, we can construct the automaton $\mathcal N_2$. Automaton $\mathcal N_2$ accepts a sequence $\sigma \in (\nats\rightarrow 3^\AP)$ if a proposition $p\in \AP$ is incorrectly mapped to $\top$ or $\bot$.
Starting with the alternating B\"uchi automaton for the formula $\varphi$, we extend the alphabet with symbols $*_\top$ and $*_\bot$ and build an automaton $\mathcal U_2^* = (\{\top,\bot,?,*_\top, *_\bot\}^\AP, Q_2^*,q_{0,2}^*,F_2^*, \delta_2^*)$. Whenever we read a symbol $v$ where some $p \in O $ is mapped to $*_\top (*_\bot)$, the automaton follows the transition for $v(p)=\bot (\top)$. After turning $\mathcal U_2^*$ to a nondeterministic automaton and projecting, a label $v$ is replaced by a label $v'$ such that for every $p \in O$, if $v(p)=*_\top$ or $v(p)=*_\bot$ then $v'(p)=\top$ or $v'(p)=\bot$, respectively. The automaton $\mathcal N_2$ is doubly-exponential in the length of $\varphi$.

\begin{proof}
Let $\sigma \in (\nats \rightarrow 3^\AP )$. We distinguish three cases:
\begin{itemize}
	\item $\sigma \in \overline{\min(\varphi)}$ and for some $i$ and some $p \in O$, the mapping $\sigma(i)(p)=?$ is wrong. We assume, w.l.o.g., that for all $\sigma' \in L(\varphi)$ with $\sigma_I=\sigma'_I$, that $\sigma'(i)(p)= \top$, and that $i$ is the first position for which $\sigma(i)(p)=?$ is wrong. A run of the automaton $\mathcal N_1$ over $\sigma$ is a sequence $r \in (\nats \rightarrow 2^{Q_1^*})$. Let $r= X_0X_1...$ be the run of the automaton $\mathcal N $ on $\sigma$, where $X_0= \{q_{0,1}^*\}$, and up to the position $i$ the run follows for each mapping to $?$ the transitions in $\mathcal N_1$ that were transitions for mappings to $?$ in the automaton $\mathcal N_1^*$ before the projection, i.e., all sets $X_j$ with $j \leq i$ contain only states $(q,1)$ from $ Q_1^*$, where $q\in Q_1^{\mathcal U}$. In the position $i$, where the mapping to $?$ is incorrect, the run follows the transition with $?$ in state $X_i$ of $\mathcal N_1$ that can be mapped to a transition $*_\bot$ in the automaton $\mathcal N_1^*$ which moves to a set $X_{i+1}$ with only states $(q,2)$ from $ Q_1^*$, i.e., the transition that checks whether replacing $?$ at $i$ with $\bot$ always leads to rejecting states for possible instantiations of upcoming $?$. As $\mathcal U_1^*$ is built from copies of the automaton $\mathcal U_1$ for the formula $\neg \varphi$, following the transition for $*_\bot$ means replacing at position $i$ the value $?$ with $\bot$, which can only lead to rejecting runs, because the automaton $\mathcal U_1$ accepts no sequence where $p$ is mapped to value $\bot$ at position $i$. 
	\item $\sigma \in \overline{\min(\varphi)}$ and for some $i$ and some $p \in O$, $\sigma(i)(p)$ is incorrectly mapped to $\top$ or to $\bot$. With the same argumentation of the last case over the structure of the automaton $\mathcal N_2$ the claim can be proven.
 	\item $\sigma \in \min(\varphi)$. In this case, for each position $i$, for each proposition $p \in O$ such that $\sigma(i)(p)=?$, and for each instantiation of $?$ for $p$ in position $i$, there are instantiations for all other $?$ values in $\sigma$ and for all propositions such that the resulting sequence $\sigma' \in (\nats \rightarrow 2^\AP)$ is in $L(\varphi)$. Let $r=X_0X_1 ...$ be a run of $\mathcal N_1$ on $\sigma$. If $r$ follows all transitions for a mapping to $?$ that correspond to a transition for the value $?$ in $\mathcal N_1^*$. In this case, all sets $X_j$ for $j\ge 0$ have states $(q,1)$ of $\mathcal U_1^*$ where $q \in Q_1^{\mathcal U}$ and the run is not accepting, because the run simulates a universal run tree in $\mathcal U_1^*$ with at least one non-accepting branch, because there is an instantiation for $\sigma$ that is a model of $\varphi$. 
       If at any point, then run $r$ takes a transition for some mapping to $?$ that corresponds to a transition $*_\bot$ or $*_\top$ in the automaton $\mathcal N_1^*$, then the run cannot be accepting, otherwise there is a mapping to $?$ for some proposition $p \in O $ in some position in $\sigma$ for which all other $?$ in $\sigma$ cannot be instantiated appropriately in order to get a model in $\sigma$.%
       
       In a similar way we can also prove that $\mathcal N_2$ has no accepting run for $\sigma$. 
\end{itemize}
\qed
\end{proof}
To check whether a skeleton $\mathcal S $ is a model for a given LTL formula $\varphi$ we compute the product $\mathcal P = \mathcal S \times \mathcal N$ where $\mathcal N$ is nondeterministic B\"uchi automaton with $L(\mathcal N) = \overline{\min(\varphi)}$ constructed in Lemma \ref{lem:NBAforLTL}. If $\mathcal P$ contains a path that simulates an accepting path in $\mathcal N$, then $\mathcal S$ has a path that violates the property $\varphi$, i.e., there is a sequence in the language $L(\mathcal S)$ that is not in $\min(\varphi)$. 

Instead of constructing the product automaton $\mathcal P$ one can also guess a run in $\mathcal P$ and check whether it is accepting\footnote{This follows the idea of the $\pspace$ model checking algorithm for LTL over transition systems \cite{Baier:2008:PMC:1373322}}. Based on this idea, the complexity of model checking skeleton is given by the following theorem. 

\begin{theorem}
	Checking whether a skeleton $\mathcal S$ is a model for an LTL formula $\varphi$ is in $\expspace$.
\label{lem:modelIn} 
\end{theorem}	

\section{Synthesis of Skeletons}

For a set of atomic propositions $\AP = I \cup O$, to check whether there is $2^O$-labeled $2^I$-transition system $\mathcal T$ that satisfies a given LTL formula $\varphi$, one would construct a deterministic $\omega$-automaton $\mathcal D$ (for example a parity automaton) with $L(\mathcal D)=L(\varphi)$, interpret the automaton as a tree automaton over trees with labels from $3^O$ and directions from $2^I$ and check its emptiness. In case, the language of the automaton is not empty the procedure returns a transition system  $\mathcal T$ that models the formula $\varphi$. In the same fashion, we can construct a deterministic $\omega$-automaton for the language $\min(\varphi)$ (for example by determinizing the automaton from Lemma \ref{lem:NBAforLTL}) and check whether there is a skeleton that is a model for  $\varphi$ by performing an emptiness check over tree automaton interpretation of the deterministic automaton.

 The deterministic automaton is very expensive to construct (triple exponential in the formula $\varphi$). Instead, we show that we can avoid this construction of the large deterministic automaton using learning. In comparison to transition systems, given an LTL formula, we show that it has a unique minimal skeleton that models the formula. The language of the skeleton is a safety language, and thus, can be characterized by a bad-prefix automaton, which is a finite word automaton. We use the learning algorithm L$^*$ to learn the deterministic bad-prefix automaton \cite{Angluin:1987:LRS:36888.36889}, which can be easily transformed to a skeleton that models the formula. The learning algorithm learns the skeleton in time polynomial in the size of the minimal skeleton.

\subsection{Learning Skeletons}

In the following we present an algorithm for learning  skeletons of LTL formulas. 
Our algorithm is based on the  $\text{L}^*$ algorithm for learning deterministic finite automata introduced by Dana Angluin \cite{Angluin:1987:LRS:36888.36889}. The setting of the L$^*$ algorithm involves two key actors, the \emph{learner} and the \emph{teacher}. The learner tries to learn a language known to the teacher by learning a minimal deterministic finite word automaton for the language. The interaction between the learner and the teacher is driven by two types of queries: \emph{membership queries}, where the learner asks whether a particular word is in the language, and \emph{equivalence queries}, to check whether a learned  deterministic finite automaton indeed defines the language to be learned. Here, the teacher responds either with a ``yes'' or with a  counterexample, which is a word in the symmetric difference of the language of the learned automaton and the actual language. A teacher is called \emph{minimally adequate}, if she can answer membership and equivalence queries.   
 
\begin{theorem}{\emph{\cite{Angluin:1987:LRS:36888.36889}}}
	Given a minimally adequate Teacher for an unknown regular language $L$, we can construct a minimal finite word automaton that accepts $L$,  in time polynomial in the number of states of the automaton and the length of the largest counterexample returned by the teacher. 
\label{theo:lstar}
\end{theorem} 
 
For an LTL formula $\varphi$ we show that the language of a skeleton that satisfies $\varphi$ is a safety language. This can be characterized by a language over finite words, namely the language of bad-prefixes.  The L$^*$ algorithm can learn a finite automaton for the language of bad-prefixes, which in turn can then be transformed to a skeleton for the property $\varphi$. 
 
\begin{lemma}
	For an LTL formula $\varphi$, the language $\min(\varphi)$ is a safety language. 
\label{lem:safetymin}
\end{lemma} 
\begin{proof}
	We show that every $\sigma \in \overline{\min(\varphi)}$ has a bad-prefix. We distinguish two cases for $\sigma$: 
	\begin{itemize}
		\item There is a point $i$ in $\sigma$ and a proposition $p$ such that $\sigma(i)(p)=\top(\text{or }\bot)$ and there is a sequence $\sigma' \in L(\varphi)$ with $\sigma_I=\sigma'_I$ and $\sigma'(i)(p)=\bot(\text{or} \top)$. Thus, any finite sequence $v_0 \ldots v_i \in (3^{\AP})^*$ with $(v_0 \ldots v_i)_I = (\sigma(0) \ldots \sigma(i))_I$ and $v_i(p)\not = ?$ is a bad-prefix for $\min(\varphi)$.
		\item There is a point $i$ in $\sigma$ and a proposition $p$ such that $\sigma(i)(p)=?$ and for all $\sigma' \in L(\varphi)$ with $\sigma_I=\sigma'_I$ we have $\sigma'(i)(p)$	 is solely $\top$ or solely $\bot$. In this case, every finite sequence $v_0 \ldots v_i \in (3^{\AP})^*$ with $(v_0 \ldots v_i)_I = (\sigma(0) \ldots \sigma(i))_I$ and $v_i(p) = ?$ is a bad-prefix for $\min(\varphi)$.
	\end{itemize}
	\qed
\end{proof}

From the last lemma we deduce, that a skeleton $\mathcal S$ for an LTL formula $\varphi$ can be seen as a safety automaton that accepts the language of minimal satisfying open sequences for $\varphi$. In particular, there is a bad-prefix automaton $\mathcal B$ that accepts the language of bad-prefixes of the  language $\min(\varphi)$.

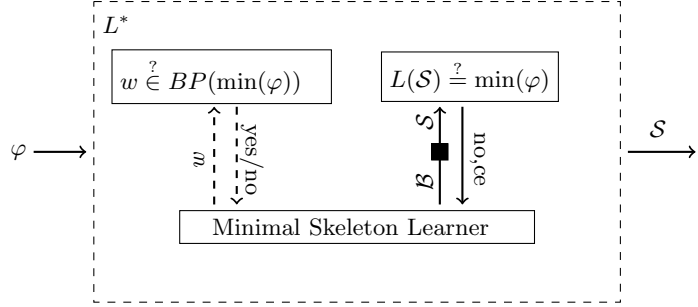
\begin{figure}[t]
\centering
\begin{tikzpicture}
	\node [] at (0,0){$\varphi$};
	\draw [thick, ->] (0.2,0) -- (0.9,0);
	\draw [dashed] (1,2) -- (8,2) --(8,-2) --(1,-2)--(1,2);
	\node [] at (1.3,1.7){$L^{*}$};
	\node [draw,text width = 2.7cm] at (2.7,1){ $w \stackrel{?}{\in} BP{(\min(\varphi))}$};
	\node [draw] at (6,1){  $L(\mathcal S) \stackrel	{?}{=} \min(\varphi)$};
	\node [draw, text width = 4.5cm] at (4.5,-1){~~~Minimal Skeleton Learner};
	\draw[thick, dashed, ->] (2.6,-0.7)--(2.6,0.6);
	\node[rotate=90] at (2.4,-.1){\emph{w}};
	\draw[thick, dashed, ->] (2.9,0.6)--(2.9,-0.7);
	\node[rotate=-90] at (3.1,-.1){yes/no };
	\draw[thick, ->] (5.6,-0.7)--(5.6,0.6);
	\node[rotate=90] at (5.4,.4){$\mathcal S$};
	\node[draw, fill=black] at (5.6,0){};
	\node[rotate=90] at (5.4,-.4){$\mathcal B$};
	\draw[thick, ->] (5.9,0.6)--(5.9,-0.7);
	\node[rotate=-90] at (6.1,-.1){{no,ce}};
	\draw [thick, ->] (8.1,0) --(9,0);
	\node [] at (8.5,0.3){$\mathcal S$}; 
\end{tikzpicture}
\caption{A modified L$^*$ for learning minimal skeletons of LTL formulas}
\label{fig:LStar}
\end{figure}

We use the L$^*$ algorithm to learn a deterministic bad-prefix automaton for the language $\min(\varphi)$.
Figure \ref{fig:LStar} shows a high level flow graph of the  learning algorithm\footnote{For more details on the L$^*$ algorithm we refer the reader to \cite{Angluin:1987:LRS:36888.36889}.}. The learner poses a series of membership questions before making a conjecture about the bad-prefix automaton. With a membership query the learner asks whether a finite word $w\in (3^{\AP})^*$ is a bad-prefix for $\min(\varphi)$. If $w$ is a bad-prefix then the teacher returns \emph{yes}, and \emph{no} otherwise. The equivalence queries allow the learner to check whether a skeleton $\mathcal S$ is correct, i.e., $L(S)=\min(\varphi)$. The teacher either confirms the automaton or returns a counterexample to the learner. The latter is either a bad-prefix that is not rejected by $\mathcal B$ or word $w\in (3^{\AP})^*$ that is not a bad-prefix for $\min(\varphi)$ yet is in the language of $\mathcal B$. The black box shown in Figure \ref{fig:LStar} between the bad-prefix automaton and a skeleton, is a check whether the safety language characterized by the bad-prefix automaton can be represented by a skeleton. We will refer to this check as the \emph{output consistency check} and will explain it later in more detail. 

The skeleton returned by the learning procedure is minimal and it is unique. 

\begin{lemma}
	For each LTL formula $\varphi$ there is a unique (up to isomorphism) minimal skeleton $\mathcal S$ such that $\mathcal S \models \varphi$.
\label{lem:uniSke} 
\end{lemma}
\begin{proof}
	Let $\mathcal S=(S,s_0,\tau,o)$ and $\mathcal S'=(S',s_0',\tau',o')$ be two minimal skeletons for $\varphi$, i.e, $|S|=|S'|=c$ and there is no  skeleton $\mathcal S''=(S'',s''_0,\tau'',o'')$ for $\varphi$ with $|S''|<c$. We show that $\mathcal S$ and $\mathcal S'$ define the same skeleton up to isomorphism. Let $\beta = \{(s,s')\in S\times S'~|~ \forall \sigma_I \in (2^I)^*.~ \tau^*(s0,\sigma_I)=s \leftrightarrow \tau'^*(s'_0,\sigma_I)=s'  \}$. The relation $\beta$ is bijective because $\tau^*$ and $\tau'^*$ are both functional and complete. Thus, there is a one-to-one mapping between the states of $\mathcal S$ and those of $\mathcal S'$, and for each $(s_1,i,s_2)\in \tau$ we have $(\beta(s_1),i,\beta(s_2))\in \tau'$. For each $(s,s')\in \beta$ it is also the case that $o(s)=o'(s')$, otherwise, there is an input sequence that distinguishes a trace in $\mathcal S$ from the corresponding one in $\mathcal S'$, which contradicts the assumption that $L(\mathcal S) = L(\mathcal S')$. This implies that $\mathcal S$ is isomorphic to $\mathcal S'$. \qed
\end{proof}

In the next sections we show how membership and equivalence queries can be solved algorithmically. 

\subsection{Membership Queries}

In this section we show that using the ideas of the automaton presented in Lemma \ref{lem:NBAforLTL} we can check whether a word is a bad-prefix in space exponential  in the length of $\varphi$.

\begin{theorem}
	Given an LTL formula $\varphi$ and a finite word $w \in (3^\AP)^*$, checking whether $w$ is a bad-prefix for $\min(\varphi)$ is in $\expspace$. 
\label{theo:expBad}
\end{theorem}

\begin{proof}
	 A finite word $w \in (3^\AP)^*$ is a bad-prefix for $\min(\varphi)$ if $w=w_0 \dots w_n$ has a prefix and there is a sequence of input values $\varsigma$ and no sequence $\sigma: \nats \rightarrow 3^\AP$ with  $\sigma_I=\varsigma$ can extend $w$ to a sequence in $w \cdot \sigma \in \min(\varphi)$. Let $\mathcal U = (\Sigma, Q,q_0,\delta,F)$ be a universal co-B\"uchi automaton such that $L(\mathcal U)= L(\neg \varphi)$. The idea is to iteratively construct a run of the automaton $\mathcal U$ and check if the run is accepting (remember that a run of $\mathcal U$ is $Q$-tree). Given the input word $w$, we first guess which position $i$ of $w$ contains a wrong mapping and compute the set of states of the run tree over $w_0 \dots w_i$ reached at this position. Then, we compute the set of states reached via choosing the transition for which the guessed position $i$ is wrong. Form here on, we guess the next input and branch universally for all valuations of the output propositions, and compute the next set of reached states. This is repeated $2^{|Q|}$ times (At latest at position $2^{|Q|}$ we reach a set of states, that was seen before and enter a loop in the run). If during the procedure a valid accepting configuration of the universal automaton was guessed, then we have  found a sequence of inputs $\varsigma$ for which no $\sigma$ with $\sigma_I=\varsigma$  extends the prefix of $w_0 \dots w_i$ to a sequence in $\min(\varphi)$. Thus, $w$ is a bad-prefix for $\min(\varphi)$. In each step we only need to remember the currently reached set of states of $\mathcal U$, and whether we have seen an accepting configuration of $\mathcal U$. Furthermore, the number of iteration can be encoded in binary and is polynomial in the size of $\mathcal U$, which in turn is exponential in the length of $\varphi$. \qed
\end{proof}

\subsection{Equivalence Queries}

We move now to equivalence queries. To check whether a skeleton is a model for a formula $\varphi$ we apply the model checking algorithm presented in Section 4. The learning algorithm first constructs a bad-prefix automaton for the language $\min(\varphi)$. We show that this automaton can be turned into a safety automaton for $\min(\varphi)$ on which we can simulate a skeleton for $\varphi$. In case we cannot simulate the skeleton on top of the safety automaton, then there is no skeleton that models the formula $\varphi$. 

\begin{lemma}
	Given a deterministic bad-prefix automaton $\mathcal B$ for a safety property $\varphi$, we can construct a deterministic safety automaton $\mathcal S$ for $\varphi$ in time linear in the size of $\mathcal B$. 
\label{lem:B2S}
\end{lemma}
\paragraph{Construction.}
Let $\mathcal B = (\Sigma, Q,q_0,F,\delta)$ be a bad-prefix automaton for some property $\varphi$ and we assume it is complete. We construct a safety automaton $\mathcal S=(\Sigma,Q',q_0',\delta')$ for $\varphi$ by first removing all states in $F$ and then by iteratively removing all resulting sink states in the automaton.

\begin{remark}
Note that if $\mathcal B$ is minimal, so is $\mathcal S$.	
\end{remark}

Before we move on to the construction we consider following fact about skeletons and the language $\min(\varphi)$ for some formula $\varphi$.  Let $\AP=O\cup I$ be the set of atomic propositions. Let $\mathcal S=(S,s_0,\tau,o)$ be a skeleton that models the formula $\varphi$. Let $\pi_1=(s_0,i_1)(s_1,i_1) \dots$ and $\pi_2=(s_0,i_1)(s_1,i_2) \dots $ be paths in $\mathcal S$ where $s_0,s_1 \in S$ and $i_1,i_2 \in I$.  Then, both sequences $\sigma_{\pi_1}=(o(s_0)\cup i_1)(o(s_1)\cup i_1) \dots$ and  $\sigma_{\pi_2}=(o(s_0)\cup i_1)(o(s_1)\cup i_2)\dots$, must be in the set $\min(\varphi)$, otherwise $\mathcal S$ is not a model of $\varphi$. This means, if the language $\min(\varphi)$ contains sequences $(o_1 \cup i_1)(o_2\cup i_1) \dots $ and $(o_1 \cup i_1)(o'_2 \cup i_2) \dots$ with $o_2 \not = o'_2$ then there is no skeleton that models $\varphi$, because $\min(\varphi) = L(\mathcal S)$ and both traces cannot be trace of the skeleton at the same time. 

\begin{definition}[Output Consistent]For a set of atomic propositions $\AP=O\cup I$, a safety automaton $\mathcal A=(3^\AP,Q,q_0,\delta)$ is output consistent, if for each state $q\in Q$ there is a unique mapping $v \in \{\bot,\top,?\}^O$ and for all transitions $(q,v',q') \in \delta$, $v'(p)=v(p)$ for all propositions $p \in O$.   
\end{definition}

\begin{lemma}
	 Given an LTL formula $\varphi$, if there is an output consistent safety automaton $\mathcal A$ for the language $\min(\varphi)$, we can transform $\mathcal A$ to skeleton $\mathcal S$ that models $\varphi$. The size of $\mathcal S$ is equal to the size of $\mathcal A$. 
\label{lem:A2SK}
\end{lemma}
\paragraph{Construction.} Let $\varphi$ be an LTL formula and let $\mathcal A= (3^\AP, Q,q_0,\delta)$ be an output consistent safety automaton for the language $\min(\varphi)$ constructed from a deterministic bad-prefix automaton as in Lemma \ref{lem:B2S}.   
Let $Q=\{q_0,q_1 \dots q_n\}$. We can construct a skeleton $\mathcal S= (S,s_0,\tau,o)$, where $ S=\{s_0, \dots, s_n\}$ and $o(s_i)= X\cap O$ for $(q_i,X,q') \in \delta$ for some $q'\in Q$, and $(s_i, Y, s_j) \in \tau$ for $Y\subseteq I$ when $(q_i, o(s_i)\cup Y, q_j) \in \delta$.    The skeleton $\mathcal S$ models $\varphi$, because it simulates the language of $\mathcal A$. 

\begin{lemma}
Given a formula $\varphi$, if an output consistent safety automaton $\mathcal A$ with $L(\mathcal A)=\min(\varphi)$ is minimal then the skeleton $\mathcal S $ extracted form $\mathcal A$ is also minimal.
\label{lem:minSke}	
\end{lemma}
\begin{proof}
	This follows from the fact that we can use the reverse of the construction presented in Lemma \ref{lem:A2SK} to construct the safety automaton from the skeleton. Assume  $\mathcal S$ was not minimal, then there is a skeleton $\mathcal S'$ with less number of states. This one, however, can be transformed backwards to a output consistent automaton of same size, which contradicts the assumption.  \qed
\end{proof}

Once we obtain a candidate skeleton, we check whether the skeleton is a model of the formula using the model checking algorithm presented in Section 4. If the skeleton is not a model, the algorithm returns a counterexample, which is a lasso-shaped trace in the candidate skeleton. As this trace must contain a bad-prefix, we can iteratively check all prefixes of the trace using membership queries until we reach the (shortest) bad-prefix. 

Using the results presented in  Theorem~\ref{lem:modelIn} (Equivalence query checking is in $\expspace$), Theorem~\ref{theo:lstar} (L$^*$ learns a minimal bad-prefix automaton in polynomial time in the size of the minimal automaton), Theorem~\ref{theo:expBad} (Membership checking is in $\expspace$), Lemma~\ref{lem:safetymin} (The language $\min(\varphi)$ can be characterized by a finite automaton), Lemma~\ref{lem:uniSke} (The minimal skeleton is unique), Lemma~\ref{lem:A2SK} (The safety automaton is a skeleton), and Lemma~\ref{lem:minSke}, we can conclude now with following theorem. 
\begin{theorem}
	Given an LTL formula $\varphi$, we can construct a  skeleton $\mathcal S$ that models $\varphi$ in time polynomial in the size of the minimal skeleton of $\varphi$. 
\end{theorem}

\section{Conclusion} 
We have presented an analysis technique for temporal specifications of
  reactive systems that identifies, on the level of individual system
  outputs over time, which parts of the implementation are determined
  by the specification, and which parts are still open. 
  Based on the algorithms developed in this paper, a synthesis tool
  can represent this information in the form of a skeleton for the
  reactive system. 
  Skeletons are more informative than conventional transition systems in
  identifying critical situations that are still underspecified.

  Our automaton-based model checking algorithm for skeletons also serves as the teaching oracle in the learning-based synthesis algorithm.
   The learning algorithm L$^*$ can be used to synthesize minimal skeletons because skeletons define safety languages, which can be characterized by a unique minimal bad-prefix automaton. Once the automaton is learned, it can directly be transformed into a skeleton for the specification.  The skeleton is minimal and can be constructed in time polynomial in the number of states of the skeleton. 

   In the development of a reactive system, skeletons can be seen as an intermediate step between the specification of the system and its implementation.
   In future work, we plan to investigate this aspect further, by exploring an incremental development process, where the refinement of the specification is guided by the identification of underspecified situations through the skeletons synthesized from the intermediate specifications.

\bibliography{biblio.bib}
\bibliographystyle{plain}

\newpage
\appendix

\end{document}